\documentclass[a4paper,english]{amsart}
\usepackage{color}
\usepackage{amsmath}
\usepackage{amssymb}
\usepackage{amsthm}
\usepackage{pgf}
\usepackage{tikz}
\usepackage[colorlinks=true,citecolor=blue]{hyperref}
\usetikzlibrary{patterns,snakes}

\title{Computability of the entropy of one-tape Turing
  machines}
\author{Emmanuel Jeandel}

\newtheorem{theorem}{Theorem}
\newtheorem{proposition}[theorem]{Proposition}
\newtheorem{fact}[theorem]{Fact}
\newtheorem{corollary}[theorem]{Corollary}
\newtheorem{definition}{Definition}
\newenvironment{sitemize}{%
\list{$\circ$\vrule height 6pt width 0pt}{\leftmargin=1cm}}{%
   \endlist
}

\begin{document}

\maketitle

\begin{abstract}

We prove that the maximum speed and the entropy of a one-tape Turing
machine are computable, in the sense that we can approximate them to
any given precision $\epsilon$. This is contrary to popular belief, as
all dynamical properties are usually undecidable for Turing machines.
The result is quite specific to one-tape Turing machines, as it is not
true anymore for two-tape Turing machines by the results of Blondel et
al., and uses the approach of crossing sequences introduced by Hennie.

\end{abstract}

\section*{Introduction}
The Turing machine is probably the most well known of all models of
computation. This particular model has many variations, that all lead
to the same notion of computability. The easiest model is the Turing
machine with just one tape and one head, that we will consider in this
paper. 

From the point of view of computability, this model is equivalent to
all others. From the point of view of \emph{complexity}, however, the
situation is very different.
Indeed, it is well known \cite{Hennie,Hartmanis, Trak2} that a language (of finite words)
accepted by such a Turing machine in \emph{linear} time is always
regular. More precisely, it can be proven that if such a Turing
machine is in time $O(n)$ on all inputs, then there is a constant $k$
so that, on any input, the machine passes at most $k$ times in any given position.

We will consider in this paper the Turing machine as a
\emph{dynamical system}: The execution is starting from \emph{any}
given configuration $c$, ie any initial state, and any initial tape, 
and we will observe the evolution.
While the Turing machine is a model of computation, it is however
quite important in the study of dynamical systems. It was intensively
studied by Kurka \cite{KurkaTuring}, and Moore \cite{Moore2, Moore} proved
that they can be embedded in various ``real life'' dynamical systems.
As an example, the uncomputability of the entropy of a Turing machine,
by Blondel et al. \cite{Blo} can be use to deduce the uncomputability
of the entropy of picewise-affine maps, proven by Koiran \cite{Koiran}
in a different way.

However, these undecidability results are usually obtained for
Turing machines with \emph{two tapes}; The basic idea is to use one
tape to simulate a given Turing machine $M$, and to control the other
tape, that will mostly do nothing computationally interesting. The \emph{computational
  complexity} of the new Turing machine will come from the first tape,
but the \emph{dynamical complexity} will come from the second tape.

There is a reason why these results use Turing machines with two
tapes. We will prove indeed that some dynamical quantities for
one-tape Turing machines are actually \emph{computable}, in the sense that
there is an algorithm that given any $\epsilon$ will give an approximation upto $\epsilon$.
The two quantities we consider are the \emph{speed} and the
\emph{entropy} of a
Turing machine. While the most theoretically important quantity is the
entropy, we will concentrate our discussion in the introduction to the
speed, which is easiest to conceive.

The speed of a Turing machine measures how fast the head goes to
infinity. Informally, the speed is greater than $\alpha$ if we can 
find a configuration $c$ for which the Turing machine is roughly 
in position $\alpha n$ after $n$ units of time. 
Note that if $\alpha$ is nonzero, this means that it takes us a time
$n/\alpha = O(n)$ to be in position $n$. Now, if we recall a previous result, a Turing machine with one tape
with running time $O(n)$ on all inputs does nothing interesting. We
will prove, using the same techniques, that this also applies to the maximum speed: If the maximum
speed is nonzero, hence the running time on \emph{some} infinite
configuration is (asymptotically) \emph{linear}, then there is a
\emph{regular} (ultimately periodic) configuration that achieves
this maximum speed.

This paper is organized as follows. In the first section, we introduce
the formal definitions of the speed and entropy of a Turing machine.
In the next section, we proceed to prove the three main theorems: The
speed and the entropy are computable, and the speed is actually a
rational number,  achieved by a ultimately periodic configuration.

\section{Definitions}
We assume the reader is familiar with Turing machines.
A (one-tape) Turing machine $M$ is a (total) map $\delta_M: Q \times \Sigma \mapsto Q \times \Sigma \times \{-1,0,1\}$
where $Q$ is a finite set called the set of states, and $\Sigma$ a
finite alphabet.

Now, the best way to see it as a dynamical system might be unorthodox
at first. The idea is to consider the Turing Machine as having a
\emph{moving tape} rather than a moving head: A
\emph{configuration} is then an element of ${\mathcal C} = Q \times
\Sigma^\mathbb{Z}$, and the map $M$ on ${\mathcal C}$ is defined as
follows: $M((q,c)) = (q',c')$ where $\delta_M(q,c(0)) = (q',a,v)$,
$c'(-v) = a$ and $c'(i) = c(i+v)$ for all $i \not= -v$.
This distinction is particularly important for the definition of the
entropy to be technically correct. However it is best to consider the
Turing machines as we are used to, and we will say ``the Turing
machine is in position $i$'' rather than ``the tape has moved $i$
positions to the right''.

\subsection*{The speed}

Given a configuration $c \in {\mathcal C}$, the \emph{speed} of $M$ on $c$
is the average number of cells that are read per unit of time.
Formally, let $s_n(c)$ be the number of \emph{different} cells read
during the first $n$ steps of the evolution of the Turing Machine $M$
on input $c$.
Note that $s_n$ is subadditive : $s_{n+m}(c) \leq s_n(c) + s_m(M^n(c))$.

\begin{definition}
	
\[ \hfill \overline{s}(c) = \limsup \frac{s_n(c)}{n}
\hspace{2cm}
\underline{s}(c) = \liminf \frac{s_n(c)}{n} \hfill \hfill \]

\end{definition}

We give two examples. 
\begin{itemize}
	\item Consider a Turing machine with two states $q_1, q_2$. On
	  $q_1$, the Turing machine goes to $q_2$ without changing the
	  position of the head. On $q_2$ the Turing machine goes right
	  and changes back to $q_1$.
	  Then $\overline{s}(c) = \underline{s}(c) = 1/2$ for all $c$.
	\item Consider a Turing machine with two states $\{L,R\}$ (for
	  Left and Right) and two symbols $\{a,b\}$.
	  In state $q$, when the machine reads a symbol $a$, it goes in the
	  direction $q$.When the machine reads a symbol $b$, it writes a
	  symbol $a$ instead and changes direction.
	  On input $c = (R,w)$ where $w$ contains only the symbol $a$, the
	  Turing machine will only go to the right, and $\overline{s}(c) =
	  \underline{s}(c) = 1$.
	  On input $c = (R,w)$ where $w$ contains only the symbol $b$, the
	  Turing machine will zigzag, and will reach the $n$-th symbol to
	  the right in time $O(n^2)$, hence will see only $O(\sqrt(n))$
	  symbols in time $n$, hence $\overline{s}(c) = \underline{s}(c) = 0$.
	  On input $c = (R,w)$ where $w$ contains $b$ only at all
	  positions $(-2)^i$, the Turing machine will have read (for $n$
	  even) $2^n +  2^{n-1}$ symbols at time $2^{n+1}+2^n-2$ and $\overline{s}(c) =
	  1/2$, but only $2^{n-1}+2^{n-2}$ at time $2^{n+1} + 2^{n-2} -2$, and
	  $\underline{s}(c) = 1/3$. This is illustrated on Figure
	  \ref{fig:laseule}.
\end{itemize}	
\begin{figure}
	\begin{center}
\begin{tikzpicture}[scale=0.07]
	\draw (0,0) -- (46,46);
	\draw[->] (-2,0) -- (-2,47);
	\draw (-2,23) node[left] {$t$};
	\draw[dashed] (0,0) -- (0,46);
	\draw[dashed] (46,0) -- (46,46);
	\draw[dashed] (-2,46) -- (46,46);
    \draw (46,46) node[above] {$n$};
	\draw (0,46) node[above] {$0$};	
    \draw (46,46) node[above] {$n$};
    \draw (-2,46) node[left] {$n$};
\end{tikzpicture}	
\begin{tikzpicture}[scale=0.07]
	\draw (0,0) -- (1,1) -- (-1,3) -- (2,6) -- (-2, 10) -- (3,15) -- (-3,21) -- (4,28) -- (-4, 36) -- (5, 45);
	\draw[->] (-14,0) -- (-14,46);
	\draw (-14,23) node[left] {$t$};
	\draw[dashed] (0,0) -- (0,45);
	\draw[dashed,thick] (5,0) -- (5,46);
    \draw (0,46) node[above] {$0$};	
    \draw (5,46) node[above] {$n$};
	\draw[dashed] (10,45) -- (-14,45) ;
	\draw (-14,46) node[left] {$n(2n-1)$};
\end{tikzpicture}	
\begin{tikzpicture}[scale=0.07]
	\draw (0,0) -- (1,1) -- (-2,4) -- (4,10) -- (-8, 22) -- (16,46);
	\draw[->] (-14,0) -- (-14,47);
	\draw (-14,23) node[left] {$t$};
	\draw[dashed] (0,0) -- (0,46);
	\draw[dashed] (-8,0) -- (-8,46);
	\draw[dashed] (16,0) -- (16,46);
	\draw[dashed,thick] (4,0) -- (4,46);
    \draw (0,46) node[above] {$0$};	
    \draw (6r,46) node[above] {$2^{n-2}$};
    \draw (16,46) node[above] {$2^n$};
    \draw (-8,46) node[above] {$2^{n-1}$};
	\draw[dashed] (20,46) -- (-14,46) ;
	\draw[dashed] (20,34) -- (-14,34) ;
	\draw (-14,46) node[left] {$3 . 2^{n} - 2$};	
	\draw (-14,34) node[left] {$9 . 2^{n-2} - 2$};
\end{tikzpicture}	

\end{center}
\caption{Three different behaviours of the same Turing machine on
  three different inputs. In the first one, the speed is $1$. In the
  second one the speed is $0$. In the third one, the speed is between
  $1/3$ and $1/2$. Time goes bottom-up}
\label{fig:laseule}
\end{figure}
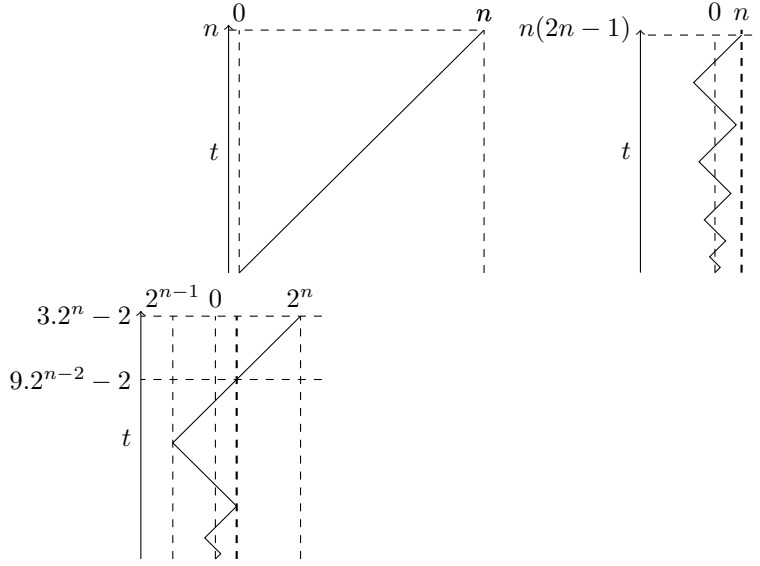
Now we define the speed of a Turing machine as the maximum of its average
speed on all configurations:
\begin{definition}
	\label{def:aprover}
\[S(M) = \max_{c \in {\mathcal C}} \underline{s}(c) =
	\max_{c \in {\mathcal C}} \overline{s}(c) =
	\lim_n \sup_c \frac{s_n(c)}{n} =  \inf_n \sup_c \frac{s_n(c)}{n}\]
\end{definition}
The fact that all these definitions are equivalent, and that the
maximum speed is indeed a maximum (it is reached by some configuration), is a
consequence of the subadditivity of $(s_n)_{n \in \mathbb{N}}$, see
\cite[Theorem 1.1]{Feng} or \cite{Mate} for a more combinatorial proof.
We give here a direct proof, following \cite{Mate} (the main difference is that a continuity argument is
replaced by a finiteness argument).

\begin{theorem}
\[S(M) = \max_{c \in {\mathcal C}} \underline{s}(c) =
	\max_{c \in {\mathcal C}} \overline{s}(c) =
	\lim_n \sup_c \frac{s_n(c)}{n} =  \inf_n \sup_c \frac{s_n(c)}{n}\]
\end{theorem}
\begin{proof}
	It is clear that the last two quantities are equal ($g_n =
	\sup_c s_n(c)$ satisfies $g_{n+m} \leq g_n + g_m$), and that they
	are larger than the others. Let $\alpha =\inf_n \sup_c \frac{s_n(c)}{n}$.
	We will prove that there exists $c$ so that $s_n(c)	\geq \alpha n$
	for all $n$, which will prove the theorem.
	By a compactness argument, it is sufficient to prove that for all
	$n$, there exists $c$ so that $s_k(c) \geq \alpha k$ for all $k \leq n$.
	(let $c_n$ be such a configuration, and take $c$ to be any limit
	point of $(c_n)_{n \geq 0}$).
	
	So suppose it is false for some $n$. That is for all
	configurations $c$, there exists $k \leq n$ so that
	$s_k(c) / k < \alpha$.
	Recall that $s_k$ is an integer between $0$ and $k$, so that the
	left side of this equation, when $k$ and $c$ varies,
	can take only finitely many values.
	This means there exists $\beta< \alpha$ so that for all
	configurations $c$, there exists $k \leq n$ so that
	$s_k(c) / k \leq \beta$.

	Now let $c$ be any configuration.
	There exists $0 < k_0 \leq n$ so that $s_{k_0}(c) \leq \beta k_0$.
	Applying the same reasoning to $M^{k_0}(c)$, we know there exists
	$0 < k_1 \leq n$ so that $s_{k_1}(M^{k_0}(c)) \leq \beta k_1$, etc.
	
    By subadditivity of $s$, this implies that for every $m$, we have $s_m(c) \leq \beta k_0
	+ \beta k_1 + {\ldots}  + s_{p} (M^{k_0 + k_1 + \dots k_q} (c))$
	for some $p \leq n$ and some $q$ so that $k_0 + \dots k_q + p = m$.
	Hence $s_m(c) \leq \beta m + n$.
	For $m = \lceil \frac{2n}{\alpha - \beta} \rceil$, this implies
	$\frac{s_m(c)}{m} \leq \frac{\alpha + \beta}{2}$. This is true of
	all $c$, so 
	$\sup_c \frac{s_m(c)}{m} \leq \frac{\alpha + \beta}{2} < \alpha$, a
	contradiction with the definition of $\alpha$.  
\end{proof}

\subsection*{The entropy}
The (topological) entropy of a Turing machine is a quantity that measures
the complexity of the trajectories. It represents roughly the average number
of bits needed to represent the trajectories.

For a configuration $c$, the \emph{trace} of $c$ is the word
$u \in (\Sigma \times Q)^\mathbb{N}$ where $u_i$ contains the letter
in position $0$ of the tape and the state at the $i$-th step
during the execution of $M$ on input $c$. We note $T(c)$ the trace of
$c$ and $T(c)_{|n}$ the first $n$ letters of the trace.
Finally, we denote by $T_n = \{ T(c)_{|n}, c \in {\mathcal C}\}$

Then the entropy can be defined by
\begin{definition}
	\label{def:entropy}
\[	H(M) =	\lim_n \frac{1}{n} \log \left| T_n \right| = \inf_n
	\frac{1}{n} \log \left| T_n \right|
\]
\end{definition}
The limit indeed exists and is equal to the infimum as $(\log
|T_n|)_{n \in \mathbb{N}}$
is subadditive. This definition is a specialized
version for (moving tape) Turing machines of the general
definition of entropy, and was proven equivalent in \cite{oprocha}.

Let's go back to the examples. In the first case, $T(c)_{|n}$ can take
roughly $|\Sigma|^{n/2}$ different values, and $H(M) = 1/2 \log |\Sigma|$.
In the second case, it can be proven that the first $n$ letters of
$T(c)$ can contain at most $\sqrt{n}$ symbols $(b,L)$ or $(b,R)$ (the
maximum is obtained starting from a configuration with only $b$). As a
consequence, $T_n$ is of size at most  $\sum_{i\leq {\sqrt{n}}}
\left(\begin{array}{c}n \\ i\end{array}\right) \leq \sqrt{n}
	\left(\begin{array}{c}n \\ \sqrt{n}\end{array}\right)$ so that $H(M) = 0$.

It is possible to give a definition for the entropy that is very similar
to the speed. For this, we use \emph{Kolmogorov complexity}. The
(prefix-free) Kolmogorov complexity $K(x)$ of a finite word $x$ is roughly
speaking the length of the shortest program that outputs $x$.

We do not define precisely the Kolmogorov complexity, see e.g.
\cite{Downey},
but will use mainly the following easy results:
\begin{itemize}
	\item For any alphabet $\Sigma$, there exists constants $c$ and $c'$ so that
	  for all words $u$ over $\Sigma$, $K(u) \leq |u|\log|\Sigma| +
	  2\log |u| + c$ and for all words $u,v$, $K(uv) \leq K(u) + K(v) + c'$.
	\item For any computable function $f$, there exists a constant $c$
	  so that $K(f(w))\leq K(w) + c$ whenever $f(w)$ is defined.
\end{itemize}

For a trace $t$, define the \emph{lower} and \emph{upper
complexity} of $t$ by $\underline{K}(t) =
\liminf \frac{K(t_{|n})}{n}$
and $\overline{K}(t) = \limsup \frac{K(t_{|n})}{n}$.

\begin{theorem}\cite{Brudno,Simpsoneff}
	\label{thm:brudno}
\[H(M) = \max_{c \in {\mathcal C}} \underline{K}(T(c)) = \max_{c \in {\mathcal
		C}} \overline{K}(T(c))\]
\end{theorem}

From this definition, it will not be surprising that we can obtain
results on both speed and entropy using the same arguments.
\clearpage
\section{A relation between the speed and the entropy}

While this section is not strictly necessary for our result, it gives
some intuition about what is happening, and an interesting relation
between the two quantities.

Recall the definitions of the speed and the entropy, in terms of
Kolmogorov complexity. It is clear that $T(c)_{|n}$ depends only
on the content of the cells visited during the $n$ first steps of the Turing machine, and on the initial
state.
In particular $K(T(c)_{|n}) \leq s_n(c) \log |\Sigma| + O(\log s_n(c))$

As a consequence, we obtain
\begin{proposition}
	\label{prop:limsup}
	For all configurations $c$,
	\[\overline{s}(c) \geq \limsup \frac{K(T(c)_{|n})}{n\log|\Sigma|}\]
	\[\underline{s}(c) \geq \liminf \frac{K(T(c)_{|n})}{n\log|\Sigma|}\]
	In particular 
	\[S(M) \geq \frac{H(M)}{\log |\Sigma|}\]
\end{proposition}
For some Turing machines (as one of the examples), the inequality may be
strict: this happens when there are some configurations of great speed,
but not \emph{many} of them.

We can overcome this problem as follows.
If $M$ is a Turing Machine over the alphabet $\Sigma$ and $A$ an other
alphabet, we denote by $M^{A}$ the Turing machine over the alphabet
$\Sigma \times A$ that works as $M$ works, without changing the $A$
component (formally, the transition function satisfies $\delta_{M^{A}}(q,(x,a)) = (q', (x',a),
v)$ if $\delta_M(q,x) = (q',x',v)$).

\begin{proposition}
\[	\frac{H(M^{A})}{\log |A|} \geq S(M) \]
\end{proposition}
\begin{proof}
First a few notations.
If $c$ is a configuration (a pair of a state and a word over $\Sigma)$ 
for the Turing Machine $M$ and $w$ is a word over the alphabet $A$, we
denote by $c \otimes w$ the configuration for the Turing Machine $M^A$
whose state is the state of $c$, and the letter in position $i$ the pair $(c_i, w_i)$.
Denote by ${\mathcal C}^{A}$ the set of configurations of $M^{A}$ and
$T^{A}$ the trace of the machine $M^{A}$.

Let $n$ be an integer and $c$ a configuration of maximal speed.
Now $T_n^{A}$ is of size at least $|A|^{s_n(c)}$, as is witnessed by
all configurations $c \otimes w$.
Hence,
\[ \frac{1}{n} \log \left| T_n^{A} \right| \geq \frac{s_n(c)}{n} \log |A|\]
By taking a limit on both sides, we get the result.
\end{proof}

We obtain the main result of this section:
\begin{corollary}
	\[ S(M) = \lim_{|A| \rightarrow \infty}	\frac{H(M^{A})}{\log
	  |\Sigma \times A|} \]
\end{corollary}	
\begin{proof}
	Remark that $S(M) = S(M^A)$ and that $\log |A| / \log |A \times
	\Sigma| \rightarrow 1$, and combine the last two propositions.
\end{proof}	

To finish, we evaluate more precisely $H(M^A)$.
For this, we look at $T_n^A$, the first $n$ bits of the trace of $M^A$.
$T_n$ is the first $n$ bits of the trace of $M$, and is obtained from
different initial configurations $c_1 \dots c_k$. Let $E_n = \{ c_1
  \dots c_k\}$ so that $T_n = \{t(c)_{|n}, c \in E_n\}$, then it is
clear that $|T_n^A| =  \sum_{c \in E_n} |A|^{s_n(c)}$, and
\[ H(M^A) = \lim \frac{1}{n} \log \sum_{c \in E_n} |A|^{s_n(c)}
	  = \lim \frac{1}{n} \log \sum_{c \in E_n} 2^{s_n(c) \log |A|}\]
	
Now this last definition makes sense also when $|A|$ is not an integer. 
If we replace $\log |A|$ by $x$ in the last definition, we obtain what
is called the (topological) \emph{pressure} \cite{Feng} of $(s_n)_{n\in \mathbb{N}}$, denoted
$P_s(x,M)$. In the context of Turing machines, the pressure
has therefore a nice interpretation, at least in the case where $x$ is the
logarithm of an integer. In particular we recover $P_s(0,M) = H(M)$
and we obtain $\lim_{x \rightarrow +\infty} P_s(x,M) / x = S(M)$,
which is also a consequence of a more general result, see \cite[Theorem 1.2]{Feng}.

\section{Computability of the speed and the entropy}

We will prove in this section that the speed and the entropy of a
Turing machine are computable.

The proof goes as follows. 
First, by the definition of the speed as an infimum, we can
compute a sequence $s_n$ so that $S(M) = \inf s_n$.
So it is sufficient to find a (computable) sequence $s'_n$ so that $S(M) = \sup
s'_n$ to compute the speed (find a $n$ so that $s_n - s'_n \leq \epsilon$)

To find such a sequence $s'_n$, it is sufficient to find
configurations $c_n$ of near maximal speed. To do that, we need to
better understand configurations of maximal speed.

First, we will establish (Propositions
\ref{prop:droite} and \ref{prop:gauche}) that a configuration of
maximal speed (entropy) cannot do too many zigzags, and must be only
finitely many times at any given position. The idea is that 
revisiting cells that were already visited is a loss of  time (and
complexity), so the machine should avoid doing it.

In the same vein, we can prove that the zigzags must not be too large
(Proposition \ref{prop:zig}): the time of the first and last visit of
a given cell must be roughly equivalent ($l_n(c) \sim f_n(c)$ in the
notation of this proposition).

All this work allows us to redefine the problem as a graph problem:
given a weighted (infinite) graph, find the path of minimum average weight
(Proposition \ref{prop:graph}). Using the graph approach, we will then prove
(Theorem \ref{thm:max}) that this average minimum weight can be well
approximated by considering only \emph{finite} graphs.
Finally, the speed and entropy for finite graphs are easy to compute
(Theorems \ref{thm:speed} and \ref{thm:entropy}), which ends the proof.

In each section, the proofs will always be done first for the speed,
then for the entropy. We deliberately choose to have similar proofs in
both cases, to help to understand the proof for the entropy,
which is more complex. In particular, some statements about the speed
are probably a bit more elaborate than they need to be.

\subsection{Biinfinite tapes are no better}

The first step in the proof is to simplify the model: we will prove
that to achieve the maximum speed (resp. maximum complexity), we only
need to consider configurations that never cross the origin, i.e.,
that stay always on the same side of the tape.
This is quite obvious, as changing from a position $i > 0$ to a
position $j < 0$ costs at least $i+(-j)$ steps, and might greatly reduce
the average speed of the TM on this configuration.

\begin{proposition}
	\label{prop:droite}
Let $c$ a configuration for which $S(M) = \lim_n \frac{s_n(c)}{n}$ and
suppose $S(M) > 0$. Then, during the computation on input $c$, 
the head of $M$ is only finitely many times in any given position $i$.
\end{proposition}
\begin{proof}
We prove only the result for $i = 0$, the result for all $i$ follows by considering $M^t(c)$ for some suitable $t$.
We suppose by contradiction that the head of $M$ is infinitely often in position $0$.

Let $k$ be an integer.
As $S(M) > 0$, there must exist a time $t$ for which the head is in position $\pm k$.
Let $t$ be the first time when this happens.
We may suppose w.l.o.g that at time $t$ the head is in position $+k$.
Now let $t'_k$ be the next time the head was in position $0$, and
finally let $t_k$ be the time at which the head was at its rightmost position
in the first $t'_k$ steps.

First, by definition $s_{t_k}(c) = s_{t'_k}(c)$. Furthermore,
$t'_k \geq t_k + s_{t_k}(c)/2$.
Indeed by definition of $t$, the leftmost position in the first $t_k$
steps is at most $-(k-1)$ so we went further to the right than to the
left in the first $t_k$ steps, so that the rightmost position is at least
in poisiton $s_{t_k}(c) /2$. Remark also that $t_k > k$ (by definition).

From this we obtain

\[
\frac{s_{t'_k}}{t'_k} \leq \frac{s_{t_k}}{t_k + s_{t_k}/2} \leq \frac{\frac{s_{t_k}}{t_k}}{1 + \frac{s_{t_k}}{2 t_k}}
\]
By taking a limsup on both sides we obtain
\[ S(M) \leq \frac{S(M)}{1 + \frac{S(M)}{2}} \]
A contradiction.
\end{proof}
\begin{proposition}
	\label{prop:gauche}
Let $c$ be a configuration for which $H(M) = \lim_n \frac{K(T(c)_{|n})}{n}$ and
suppose $H(M) > 0$. Then
for any position $i$, the head of $M$ is only finitely many times in position $i$.
\end{proposition}
\begin{proof}
	It's exactly the same proof.
	
Note that $K(T(c)_{t'_k}) \leq K(T(c)_{t_k}) + O(\log t'_k)$
(The first $t'_k$ bits of $T(c)$ can be recovered if we know only the
first $t_k$ bits, and the number of bits we want to recover), and
$t'_k \geq t_k + K(T(c)_{t_k}) / (2 \log |\Sigma|) + O(\log t_k)$
(Indeed $K(T(c)_{t_k}) \leq n \log |\Sigma| + O(\log t_k)$ 
where $n = s_{t_k}(c)$ is the number of bits read during times $t \leq
t_k$, and $t'_k \geq t_k + n/2$), from which we get the
same contradiction.
\end{proof}

These two propositions state that we only have to deal with 
configurations that never reach the position $i = 0$ once they leave
it at $t = 0$ (replace $c$ by $M^p(c)$ for a suitable $p$).

If we deal with the disjoint union of the Turing machine $M$ and its
mirror (exchange left and right) $\tilde M$, we may now assume, and
we do in the rest of this section, that the maximum speed and
complexity is reached with a configuration that never go to negative
positions $i < 0$ and, if $S(M) > 0$ (resp. $H(M) > 0$), that pass
only finitely many times to any given position.

\subsection{A reformulation}
Recall that we suppose in the following sections 
that the maximum speed is obtained for a configuration that never go
to negative positions.

Let call $f_n(c)$ ($f$ for first) the first time we reach position $n$.
Then the average speed on a configuration $c$ (for which the Turing
machine never goes in negative positions) can be defined equivalently as $\lim_n \frac{n}{f_n(c)}$.
We prove now a stronger statement.

Let's call $l_n(c)$ the \emph{last} time we reach position $n$.
If we never reach position $\pm n$, or if we reach it infinitely
often, let $l_n(c) = \infty$,

\begin{proposition}
\label{prop:zig}
\[	S(M) = \max_c \limsup \frac{n}{l_n(c)} = \max_c \liminf	\frac{n}{l_n(c)}\]
\[	H(M) = \max_c \limsup \frac{K(c_{|n})}{l_n(c)} = \max_c \liminf \frac{K(c_{|n})}{l_n(c)}\]

If the speed (resp. entropy) is nonzero, the maximum is reached
for some configuration $c$ for which  $l_n(c)$ is never infinite.
In particular, for this configuration, $l_n(c) \sim f_n(c)$
\end{proposition}
\begin{proof}
	It is clear that $S(M)$ and $H(M$) are upper bounds, as
	$n \leq s_{f_n(c)}(c)$ and $K(c_{|n}) \leq K(T(c)_{|f_n(c)}) + O(\log n)$.
	In particular the result is true if $S(M) = 0$ (resp. $H(M) = 0$).

We first deal with the speed. Let $c$ be a configuration of maximum speed.
By the previous subsection, we may suppose that $c$ never reaches
negative positions.

Let $t_n = l_n(c)$.
Let $p$ be the rightmost position the head reaches before $t_n$ and
$t'_n$ the first time this position is reached.
Note that $s_{t_n}(c) = s_{t'_n}(c) = p$ (no negative position is ever reached)

From this we get $\lim \frac{t_n}{t'_n} = \lim \frac{t_n}{s_{t_n}(c)}\frac{s_{t'_n}(c)}{t'_n} = 1$

Note also that $t'_n \geq n$ and $t_n \geq t'_n + s_{t_n} - n$ (we are at position $s_{t'_n} = s_{t_n}$ at time $t'_n$ and at position $n$ at time $t_n$)

Hence
\[ \frac{n}{t_n} \geq \frac{t'_n - t_n}{t_n} + \frac{s_{t_n}}{t_n}
\]
From which the result follows.

For the entropy, the proof is almost the same. 
From  $K(T(c)_{t_n}) = K(T(c)_{t'_n}) + O(\log t_n)$, we get again
that $\lim_n \frac{t'_n}{t_n} = 1$.

Now $K(T(c)_{t_n}) \leq K(c_n) + (t_n - t'_n) \log |\Sigma|  +O(\log
t_n)$ (the first $t_n$ bits of $T(c)$ can be recovered if we know
$t_n$ and the
first $p$ bits of $c$, hence if we know the first $n$ bits of $c$ and
the $p-n \geq t_n - t'_n$ next bits), from which the result follows
again.
\end{proof}

\subsection{Crossing sequences}
Now the last technical tool we need are \emph{crossing
  sequences}, introduced by Hennie \cite{Hennie}.

First denote by ${\mathcal C}^+$ the set of configurations $c$ on which:
\begin{itemize}
    \item The Turing machine never reaches any positions $i < 0$
	\item The Turing machine never reaches the position $0$ again once
	  it leaves it at $t = 0$.
	\item For any $i > 0$, the head of the Turing is only finitely many
	  times in position $i$.
\end{itemize}	
The last sections prove that we only have to deal with configurations
in ${\mathcal C}^+$.

Let $c$ be such a configuration. The \emph{crossing sequence} at boundary $i$ is the sequences of states of the
machine when its head cross the boundary between the $i$-th cell and
the $i+1$-th cell. We denote by $C_i(c)$ the crossing sequence at
boundary $i$. Note that $C_0(c)$ consists of a
single state, which is the initial state of $c$  (the machine never
reaches the position $0$ anymore) and $C_i(c)$ is finite for $i > 0$.

The main idea of the crossing sequences is that $C_i(c)$ represents
all the exchange of information between the positions $j \leq i$ and
the positions $j > i$ of the tape.
In particular, if $C_i(c) = C_j(c')$ for two configurations $c,c'$,
and if we consider the configuration $\tilde c$ that is equal
to $c$ upto $i$ then equal to $c'$ (shifted by $i-j$ so that the
$j+1$-th cell of $c'$ becomes the $i+1$-th cell of $\tilde c$), then the Turing
machine on $\tilde c$ will behave exactly like $c$ on all positions less
than $i$, and as $c'$ (shifted) on positions more than $i$.
Hence the crossing sequences capture exactly the behaviour of the
Turing machine.

We now consider the following labeled graph (automaton) $G$: The vertices
of $G$ are all finite words over the alphabet $Q$ (all possible
crossing sequences), and there is an edge 
from $w$ to $w'$ labeled by $a \in \Sigma$ if $w$ and $w'$ are
\emph{compatible}, in the sense that it
looks possible to find a configuration and a position $i$ so that
$C_{i}(c) = w$, $C_{i+1}(c) = w'$ and $a$ is the letter at position
$i+1$ in $c$ (said otherwise,
$w$ and $w'$ are two consecutive crossing sequences for some
configuration $c$). The exact definition is as follows. We define recursively two subsets $L$ and
$R$ of $Q^* \times Q^* \times \Sigma$ as follows:
\begin{sitemize}
	\item $(\epsilon, \epsilon, a) \in L$, $(\epsilon, \epsilon, a) \in R$
	\item If $\delta(q_1, a) = (q_2, b, -1)$ then $(q_1 q_2 w, w', a) \in L$ iff $(w,w',b)\in L $
	\item If $\delta(q_1, a) = (q_2, b, +1)$ then $(q_1w, q_2w', a)\in L$ iff $(w,w',b)\in R $
	\item If $\delta(q_1, a) = (q_2, b, -1)$ then $(q_2w, q_1w',
	a) \in R$ iff $(w,w',b)\in L $
	\item If $\delta(q_1, a) = (q_2, b, +1)$ then $(w, q_1q_2w', a)  \in R$ iff $(w,w',b)\in R $
\end{sitemize}	
Then there is an edge from $w$ to $w'$ labeled $a$ if and only if
$(w,w',a)\in L$.

Note that this echoes a similar definition for two-way finite automata
given in \cite[2.6]{HopcroftUllmann} where $(w,w',a)\in L$ is called
``$w$ left-matches $w'$'' (The note in Example~2.15 is particularly
relevant). The exact definition above is also hinted at
in  \cite{Pighizzini}.

Let us explain briefly these conditions. Suppose $\delta(q_1, a) =
(q_2, b, +1)$, and suppose that the Turing machine at some point
arrives in some cell $i$ from the left, in the state $q_1$ and sees $a$.
Then by definition, the first symbol from $C_i(c)$ must be $q_1$.
By definition of the local rule $\delta$, the Turing machine will
enter state $q_2$ and go right so that the first symbol in
$C_{i+1}(c)$ will be $q_2$. Now, the next time the Turing machine will
come into the cell $i$, it must be coming from the right, and when it does
it will see the symbol $b$. This explains the rule 
$(q_1w, q_2w', a)\in L$ iff $(w,w',b)\in R $, where $w$ and $w'$
represent the crossing sequences after the second time the Turing
machines comes to the cell $i$.

Now it is clear that a configuration $c$ defines a path in this graph
$G$, and that we can recover the speed of the configuration from the
graph, as explained in the following.

A \emph{path} in the graph $G$ is a sequence $p = \{(w_i, u_i)\}_{i < N}$ where $w_i$
is a vertex of $G$ and $u_i$ a letter from $\Sigma$ so
that  $(w_i, w_{i+1},u_{i}) \in L$ for all $i < N-1$.
A \emph{valid} path is an infinite path ($N = \infty$) so that $w_0$
consists of one single letter (state).  We denote by ${\mathcal P}(G)$ the
set of valid paths of a graph $G$.

The following facts are obvious:
\begin{fact}
	For any $c \in {\mathcal C}^+$, $\{(C_i(c), c_i)\}_{i \geq 0}$ is
	a valid path in $G$.
	
	Furthermore, for any valid path $p = \{(w_i, u_i)\}_{i \geq 0}$, there exists a
	configuration $c \in {\mathcal C}^+$ so that $u_i = c_i$ and $C_i(c)$ is
	a prefix of $w_i$.
\end{fact}
Note that it is indeed possible for $w_i$ to be strictly larger than $C_i(c)$.

Now we explain how we can redefine the speed on the graph $G$.

If $p$ is a finite path ($N$ is finite), the \emph{length} of $p$ is
$|p| = N$, the \emph{weight} of $p$ is $weight(p) = \sum_{i < N} |w_i|$,
and the \emph{complexity} of $p$ is $K(p) = K(u_0 \dots u_{N-1})$

If $p=(u_i, w_i)_{i \geq 0}$ is an infinite path, and $p_{|n} =
(u_i,w_i)_{i \leq n}$, the \emph{average speed} of $p$ is
$\underline{s}(p) = \liminf \frac{|p_{|n}|}{weight(p_{|n})}$ and the \emph{average
  complexity} of $p$ is $\underline{K}(p) = \liminf \frac{K(p_{|n})}{weight(p_{|n})}$.
We define similarly $\overline{s}(p)$ and $\overline{K}(p)$.

Now note that $\sum_{i < n} |C_i(c)|$ is bounded from below by the first
time we go to the position $n$, and from above by the last time we go
to position $n$. So by the previous section
\begin{proposition}
	\label{prop:graph}
\[	S(M) = \max_{p \in {\mathcal P}(G)} \underline{s}(p) = \max_{p \in
	  {\mathcal P}(G)} \overline{s}(p)\]
\[
	H(M) = \max_{p \in {\mathcal P}(G)} \underline{K}(p) = \max_{p \in {\mathcal P}(G)} \overline{K}(p)\]
\end{proposition}

Now to obtain the main theorems, let $G_k$ be the subgraph of $G$
obtained by taking only the vertices of size $|w_i| \leq k$.

\begin{theorem}
	\label{thm:max}
\[	S(M) = \sup_k \sup_{p \in {\mathcal P}(G_k)} \overline{s}(p) \]
\[	H(M) = \sup_k \sup_{p \in {\mathcal P}(G_k)} \overline{K}(p)\]
\end{theorem}
This means we only have to consider finite graphs to compute the speed
(resp. entropy). We will prove in the next section that the speed and
the entropy are computable for finite graphs, which will give the result.

Before going to the proof, some intuition.
Let $p$ be a path of maximum speed $S(M) > 0$.
For the speed to be nonzero, vertices of large weight cannot be too
frequent in $p$. Now the idea is to \emph{bypass} these vertices (by
using other paths in the graph $G$) to obtain a new path $p'$ with
almost the same average speed. For the speed, it's actually possible
to obtain a path $p'$ of the \emph{same} speed (this will be done in
the next section). However, for the entropy, it is likely that these
paths were actually of great complexity so that their removal gives us
a path of smallest (yet very near) average complexity.

\begin{proof}
First, the speed. One direction is obvious by definition. 
We suppose that $S(M) >  0$, otherwise the result is trivial.
Let $p$ be a path of maximum speed.

Let $k$ be any integer so that $1 /k < S(M)$. For any vertex $w$ and
$w'$ of size less or equal to  $k$ so that $p$ goes through $w$ and
$w'$ in that order, choose some finite path $P(w,w')$ from $w$ to $w'$.
Now let $K$ be an upper bound on the weights of all those paths.

The idea is now simple: we will change $p$ so that it will not go
through any vertices $w$ of size $|w| > K$.

We do it like this: Whenever there is a vertex $\tilde w$ of size greater than
$K$, we will look at the last vertex $w$ before it of size less or equal
to $k$, and to the first vertex $w'$ after it of size less or equal
to $k$, and we will replace the portion of this path by $P(w,w')$.
Let's call $p'$ this new path.
Note that there must exist such a vertex $w'$, otherwise all vertices
will be of size greater than $k$ after some time, which means the
speed on $p$ is less than $1/k$, a contradiction. 

Now we prove this construction works.

Let $n$ be an integer so that for all $m \geq n$.
\[ \frac{m}{weight(p_{|m})} \geq S(M)/2\]

Now let $m$ so that the vertex $w_m$ of $p$ is of size less than $k$.
We will look at how the $m$ first positions of $p$ where changed into $p'$.
Let $m'$ be the position of the vertex $w_m$ in $p'$ ($w_m$ still
appears in $p'$ as we only change vertices of size greater than $k$).

By the above inequality, it is clear that in the $m$ first position of
the path $p$, there is at most $2m/(kS(M))$ vertices of size greater
than $k$.
This means that this portion of the path in $p'$ is of length at least 
$m' \geq m - 2m/(kS(M))$.
Furthermore, at each time, we replace a finite path by a path of
smallest weight (each path was of weight at least $K$, and each new one
is of weight at most $K$).

As a consequence, for this new path $p'$ we have
\[ \frac{m'}{weight(p'_{|m'})} \geq \frac{m - 2m/(kS(M))}{weight(p_{|m})}\]
Hence
\[ \overline{s}(p') \geq S(M) - 2/k\]
We have proven that some path in $G_K$ is at least $2/k$ to the
optimal speed, which proves the result.

The proof for the entropy is, as always, very similar. 
We start from $1/k  < H(M)/(\log|\Sigma|))$, which guarantees that
infinitely many vertices are of weight less than $k$.
As before, we will choose $K$
greater than all weights, but now also greater than $k^2$.

First, $K(p_{|n}) \leq n\log|\Sigma| + O(\log n)$, so 
$\underline{K}(p) \leq \underline{s}(p) \log |\Sigma|$, 
so we may choose $m$ so that  for all $m\geq n$
\[ \frac{m}{weight(p_{|m})} \geq H(M)/(2\log|\Sigma|)\]

To simplify notations, let $L = H(M)/(2\log|\Sigma|))$.

We now have to evaluate $K(p'_{|m'})$.
To recover $p_{|m}$ from $p'_{|m'}$, we only need:
\begin{sitemize}
	\item To know $m$ 
	\item To know the positions of the paths that were cut.
    \item To know the labels of what was cut
	\item To know the labels of what was added instead
\end{sitemize}
First, the positions. There are at most $mL/K$ vertices of size at least $K$, 
so we did at most $mL/K$ cuts.
The cuts can be described by two sets: the set of beginnings of the
cuts, and the set of endings of the cuts.
Each set is of size at most $mL/K$. For a given size $p \leq mL/K$, there
are at most $\left(\begin{array}{c} m \\ mL/K \end{array}\right)$  sets of size $p$, so this can be
described by a binary word of size at most the logarithm of this
quantity (upto a factor $O(\log m)$).

Second, what was cut. We only cut vertices of size at least $k$, so
this can be described by a single word of size
at most $\left\lceil mL/k)\right\rceil\left\lceil \log |\Sigma|\right\rceil$.

Third, what was added. At each cut, we added one of the paths
$P(w,w')$. Now there are at most $|Q|^{k+1}$ words of size at most
$k$, hence are at most $\left(|Q|^{(k+1)}\right)^2$ such paths, and there are at
most $mL/K$ cuts, so this can be
described by a word of size at most $\left\lceil mL/K\right\rceil 2(k+1) \left\lceil \log |Q|\right\rceil$

Thus

\[
K(p'_{|m'}) \geq K(p_m) -
2 \log \left(\begin{array}{c} m \\ mL/K \end{array}\right)
- \left\lceil mL/k\right\rceil\left\lceil
\log |\Sigma|\right\rceil 
- \left\lceil mL/K\right\rceil 2(k+1) \left\lceil \log |Q|\right\rceil
- O(\log m)
\]

Now $weight(p'_{|m'}) \leq weight(p_{|m})$ and $weight(p'_{|m'}) \geq
m' \geq m (1 -L/k)$

\[
\frac{K(p'_{|m'})}{weight(p'_{|m'})} \geq \frac{K(p_m)}{weight(p_m)} -
\frac{\text{the same quantity}}{m (1-L/k)}
\]
Hence

\[
\overline{K}(p')\geq H(M) - \frac{1}{1 - L/k} \left(
-2 E\left(\frac{L}{K}\right) - \frac{L\left\lceil\log|\Sigma|\right\rceil}{k}
- \frac{L (2k+1)\log|Q|}{K}
\right)
\]
where $E(p) = -p\log p - (1-p) \log (1-p)$.
Now the quantity to the right tends to $H(M)$ when $k$ tends to
infinity (recall that we choose $K$ greater than $k^2$, which we need for the last
term), which proves the result.
\end{proof}

\subsection{The main theorems}

Now we can explain how to use the last result to prove the main theorems.

The idea is that we can compute the speed $S(M)$ from above, by the
formula
\[S(M) = \inf_n \sup_c \frac{s_n(c)}{n}\]
So it is sufficient to explain how to compute it from below and this
comes precisely from the previous theorem.

\begin{theorem}
	\label{thm:speed}
	There exists an algorithm that, given a Turing machine $M$ and a
	precision $\epsilon$, computes $S(M)$ to a precision $\epsilon$.
\end{theorem}	
\begin{proof}
We only have to explain how to compute the maximum speed for a finite
graph $G$. First, we may trim $G$ so that all vertices are reachable
from a vertex of size $1$.
It is then obvious that the maximum speed is obtained by a path that
goes to then follow a cycle of minimum average weight, so the maximum
speed is exactly the inverse of the minimum average weight.
This is easily computable, see \cite{Karp} for an efficient algorithm.
\end{proof}	
	
We can say a bit more
\begin{theorem}
  The maximum speed of a Turing machine $S(M)$ is a rational number.
  It is reached by a configuration which is ultimately periodic.
\end{theorem}
\begin{proof}
	We suppose that $S(M) > 0$ otherwise the result is clear.
	
	We will prove that the sequence $\sup_{p \in G_k} \overline{s}(p)$
	is stationary.
	
	Let $k = 1+ \lceil  1/S(M) \rceil$. Let $K = k(k+1) |Q|^{k+1}$
	
	Now we look at $\sup_{p \in G_{K'}} \overline{s}(p)$ for some $K' \geq K$.
	The maximum is reached for some path that reach some cycle of
	minimum average weight.
	
	Note that this cycle cannot be of length greater than $(k+1)|Q|^{k+1}$.
	Indeed, denote by $m$ the length of this cycle. As there are at
	most $|Q|^{k+1}$ vertices in this cycle of length at most $k$,
	the average speed on this cycle is less than
	
\[\frac{m}{(k+1) (m - |Q|^{k+1})} \leq 1/k < S(M)\]

    Now, there cannot be any vertices in this cycle of length at least 
	$k(k+1)|Q|^{k+1}$. otherwise the average speed would be less than
	
	\[ \frac{(k+1)|Q|^{k+1}}{k(k+1)|Q|^{k+1}} \leq 1/k < S(M)\]

	Hence this cycle is already in $G_K$.
	
	Now if we look at the cycle of minimal average weight in $G_K$ that can be
	reached in $G$, hence in $G_P$ from some $P$, then it is clear
	that $S(M)$ is exactly the inverse of the average weight of this
	cycle, and it is reached for some path $p$ in $G_P$ that reaches
then follows this cycle.
\end{proof}	
Note that, while the maximum speed is a rational number,
there is no algorithm that actually computes this rational number (we
are only able to approximate it up to any given precision). This can
be proven by an adaptation of the proof of the undecidability of the
existence of a periodic configuration in a Turing machine \cite{KariOll}.

Now we do the same for the entropy:
\begin{theorem}
	\label{thm:entropy}

	There exists an algorithm that, given a Turing machine $M$ and a
	precision $\epsilon$, computes $H(M)$ to a precision $\epsilon$.
\end{theorem}	
\begin{proof}
We only have to explain how to compute the maximum complexity for a finite
graph $G$. However, we do not know how to do this in the whole
generality.
We will only prove how to do it for the graphs $G_k$, that have an
additional property: It is easy to see that they are
weakly-deterministic, in the sense, that given two vertices $w$,$w'$ and 
a word $u$, there is at most one path from $w$ to $w'$.

First we trim $G_k$ so that any vertex of $G_k$ is reachable from a
vertex of size $1$.

For a given $k$, we consider a set $B_k$ of infinite words over the
alphabet $(Q \times \Sigma) \cup Q$ defined as follows: A word is in $B_k
$ if and only if it never contains more than $k-1$ consecutive letters in $Q$,
never more than $1$ consecutive letters in $Q \times \Sigma $, and all factors
of the form $(a,q)w(b,q')w'(c,q')$ 
satisfy than there is a vertex from $qw$ to $q'w'$ labeled by $b$.

Now it it clear that if $p = \{(u_i, q_iw_i)\}_{i \geq 0}$ is an infinite
path in $G_k$, then $(u_0,q_0)w_0(u_1,q_1)w_1 \dots$ is a word of $B_k$.
Conversely, any word of $B_k$, upto the deletion of at most $k+1$
letters at its beginning, represents a path in $G_k$.

Moreover, $K((u_0,q_0)w_0 \dots (u_n,q_n)w_n) = K(u_0 \dots u_n) + O(1) = K(p_{|n}) + O(1)$.
Indeed, we can recover all the states knowing only $w_0$ and $w_n$, as
the graph is weakly deterministic.
Furthermore, the length of $(u_0,q_0)w_0 \dots (u_n,q_n)w_n$ is
exactly $weight(p_{|n})$.

This means that the maximum complexity on the graph $G_k$ can be
computed as:
\[ \sup_{w \in B_k}\lim\sup \frac{K(w_0 \dots w_n)}{n} \] 

And this we know how to compute. Indeed, $B_k$ is what is called a
subshift of finite type (it is defined by a finite set of forbidden
words), for which the above quantity is exactly the entropy (!) of
$B_k$ \cite{Brudno, Simpsoneff}, which is easy to compute, see e.g.,
\cite{LindMarcus}.

To better understand what we did in this theorem, the intuition is as
follows: Computing the entropy of the trace is difficult, but the
trace can be approximated by taking only into account configurations
for which we cross at most $k$ times the frontier between any two
consecutive cells. For this approximation $T_k$ of the trace, we can reorder the letters
inside the trace so that transitions corresponding to the same
position are consecutive, and this does not change the entropy.
However, it makes it easier to compute.
\end{proof}	

\section*{Open Problems}

The main open problem is of course to strengthen the last theorem, and
actually characterise the exact numbers that can arise as entropies of
Turing machines. It cannot be all nonnegative computable numbers, as an
enumeration of Turing machines would give us an enumeration of these
numbers, which is impossible by an easy diagonalisation argument.
The natural conjecture is that the supremum in the theorem is actually
reached, which would prove that the numbers that arises as entropies of
Turing machines are exactly the numbers that arises as entropies of
subshifts of finite type, which are well known.

Finally, the situation for Turing machines with two-tapes is not
clear. Of course, we know that the speed (resp. entropy) are not
computable \cite{Blo} (there is no algorithm that given a Turing
machine and a precision $\epsilon$ computes the speed upto
$\epsilon$), but we know of no example where the speed (resp. the
entropy) is not a rational number (resp. a computable real number).

\end{document}